\documentclass[12pt]{article}
\usepackage{paralist,amsmath,float, amsthm,amssymb,mathtools,url,graphicx,pdfpages,tikz,pdfpages,rotating}
\usetikzlibrary{arrows,automata}
\usepackage[affil-it]{authblk}
\usepackage[left=1in,right=1in,top=1.2in, bottom=1in]{geometry}
\usepackage{mathdots}
\usepackage{comment}
\usepackage{bm} %to use bold math

\newtheorem{theorem}{Theorem}[section]
\newtheorem{lemma}[theorem]{Lemma}
\newtheorem{corollary}[theorem]{Corollary}

\theoremstyle{definition}
\newtheorem{definition}[theorem]{Definition}
\newtheorem{example}[theorem]{Example}

\newtheorem{remark}[theorem]{Remark}

\newtheorem{question}[theorem]{Question}

\numberwithin{equation}{section}

\renewcommand{\d}{\rm{d} }

\DeclareMathOperator{\e}{E}
\DeclareMathOperator{\ns}{NS}

\DeclareMathOperator{\rs}{RS}
\DeclareMathOperator{\supp}{supp}

\DeclareMathOperator{\eonv}{eonv}

\newcommand{\F}{\mathbb{F}}

\begin{document}
\title{Codes from incidence matrices of hypergraphs}

\author{Sudipta Mallik } 
\affil{\small Department of Mathematics and Statistics, Northern Arizona University, 801 S. Osborne Dr.\\ PO Box: 5717, Flagstaff, AZ 86011, USA  sudipta.mallik@nau.edu}
\author{Bahattin Yildiz } 
\affil{\small Intel Corporation\footnote{© Intel Corporation.  Intel, the Intel logo, and other Intel marks are trademarks of Intel Corporation or its subsidiaries.  Other names and brands may be claimed as the property of others.   }, 2200 Mission College Blvd, Santa Clara, CA 95054, USA bahattin.yildiz@intel.com}

\maketitle
\begin{abstract}
Binary codes are constructed from incidence matrices of hypergraphs. A combinatroial description is given for the minimum distances of such codes via a combinatorial tool called ``eonv". This combinatorial approach provides a faster alternative method of finding the minimum distance, which is known to be a hard problem. This is demonstrated on several classes of codes from hypergraphs. Moreover, self-duality and self-orthogonality conditions are also studied through hypergraphs.  
\end{abstract}

\renewcommand{\thefootnote}{\fnsymbol{footnote}} 
\footnotetext{\emph{2020 Mathematics Subject Classification: 94B05, 94B25\\  %Linear codes, general;; Combinatorial codes
}}
 \renewcommand{\thefootnote}{\arabic{footnote}} 

\section{Introduction}
A {\it hypergraph} $H$ is an ordered pair $H=(V,E)$ consisting of a finite nonempty set $V$, called the {\it vertex set}, and  a finite multiset $E$, called the {\it edge set}, which is a family of nonempty subsets of $V$.  We call $H$ to be {\it simple} if $H$ has  no repeated edges (i.e., $E$ is a set). $H$ is {\it $r$-uniform} if each edge of $H$ has cardinality $r$. Note that simple 2-uniform hypergraphs are just simple graphs. The {\it order} and {\it size} of $H$ are $|V|$ and $|E|$ respectively. Given $V=\{v_1,v_2,\ldots,v_n\}$ and $E=\{e_1,e_2,\ldots,e_m\}$, the {\it incidence matrix} $M=[m_{ij}]$ of $H$ is an $n\times m$ binary matrix such that $m_{ij}=1$ if and only if $v_i\in e_j$. For simplicity, index sets of rows and columns of $M$ are denoted by corresponding sets of vertices and edges of $H$ respectively.\\

Let $\F_2$ be the binary field. A binary linear code $C$ of length $n$ is defined as a subspace of $\F_2^n$. If the dimension of $C$ is $k$, we say $C$ is an $[n,k]$-code. A matrix whose rows form a basis for $C$ is called a {\it generator matrix} for $C$ and is denoted by $G$. The code $C$ is the row space of $G$ and hence can be denoted as $\rs(G)$. 

The support of a vector $\bm x=[x_1,x_2,\ldots,x_n] \in \F_2^n$, denoted by $\supp(\bm x)$, is the set of indices $i$ for which $x_i\neq 0$. The {\it Hamming weight} $w_H(\bm x)$ of a vector $\bm x \in \F_2^n$ is defined as the number of non-zero coordinates in $\bm x$, i.e., $w_H(\bm x)=|\supp(\bm x)|$. The {\it  Hamming distance} between two vectors $\bm x$ and $\bm y$ in $\F_2^n$, denoted by $d_H(\bm x,\bm y)$, is defined as $$d_H(\bm x,\bm y) = w_H(\bm x-\bm y).$$
The {\it minimum distance} of a code $C$, denoted by $d(C)$, is defined to be the minimum distance between distinct codewords in $C$. We write the standard parameters $[n,k,d]$ to describe a code $C$ where $n$ denotes the length of $C$, $k$ its dimension, and $d$ its minimum distance. 

\begin{definition}
Let $C$ be a binary linear code of length $n$. The
{\it dual} of $C$, denoted by $C^{\perp}$, is given by
$$C^{\perp}:= \large\{ \bm y \in \F_2^n \large \;|\; \langle \bm y,\bm x \rangle = 0 \:\: \forall\:
\bm x \in C \large \}.$$
\end{definition}

 Note that, if $C$ is $\rs(G)$, then $C^{\perp}$ is $\ns(G)$, the null space of $G$.

\begin{definition}
A binary linear code $C$ is {\it self-orthogonal} if $C \subseteq C^{\perp}$ and  {\it self-dual} if $C = C^{\perp}.$
\end{definition}

A graph-theoretic approach to codes was given in \cite{Us1} and \cite{Us2}. In \cite{Us1}, codes were generated using the adjacency matrices of simple graphs and a combinatorial parameter called ``von" was introduced to find the minimum distances of such codes. This parameter was used to find a new formula for the minimum distance of an expander code in \cite{SM1}. In \cite{Us2}, we considered codes generated by adjacency matrices of directed graphs and found graph-theoretic conditions on equivalence of such codes as well as their minimum distances. 

In this work, we consider the incidence matrices of hypergraphs as generator matrices for binary linear codes. Note that this covers all binary linear codes as any generator matrix can be considered as the incidence matrix of a hypergraph. Codes from incidence matrices of graphs have been studied in \cite{dankelmann1} and \cite{dankelmann2}. Our work on hypergraphs generalizes these. We also introduce the notion of ``eonv", a combinatorial parameter for hypergraphs that gives us an alternative description for the minimum distance of codes obtained from the incidence matrices of hypergraphs. This combinatorial approach leads to a more simplified calculation for minimum distances of some classes of codes. We also consider conditions on self-orthogonality and self-duality for such codes through the hypergraph structure. The rest of the work is organized as follows. In Section 2, we give the main results on eonv and apply them to special cases of hypergraphs, namely the projective planes and complete 3-partite 3-uniform hypergraphs. We also present a main result about the minimum distances of the class of cyclic codes via the eonv construction. In Section 3, we investigate self-orthogonality and self-duality conditions in terms of the hypergraph properties and posed some open problems.

\section{Main results}
\begin{definition}
Let $H=(V,E)$ be a hypergraph. For a nonempty set $S\subseteq V$, $\eonv(S)$ denotes the set of edges with odd number of vertices in $S$, i.e.,
$$\eonv(S)=\{e\in E \;:\; |e\cap S| \text{ is odd}\}.$$
\end{definition}

\begin{example}
Consider the complete 3-partite 3-uniform hypergraph with vertex set $V=\{x,y_1,y_2,z_1,z_2\}$ and edge set $E=\{\{x,y_1,z_1\},\{x,y_1,z_2\},\{x,y_2,z_1\},\{x,y_2,z_2\}\}$. 

\[\begin{array}{|c|c|}
\hline
    S & \eonv(S) \\
    \hline
    \{y_1\} & \{\{x,y_1,z_1\},\{x,y_1,z_2\}\}\\
    \hline
    \{x\} & E\\
    \hline
    \{y_1,z_1\} & \{\{x,y_1,z_2\},\{x,y_2,z_1\}\}\\
    \hline
    \{x,y_i\} & \varnothing\\
    \hline
    \{x,z_i\} & \varnothing\\
    \hline
\end{array}\]
\end{example}

Note that if $H=(V,E)$ is a $2$-uniform connected hypergraph (i.e., a connected graph), then a nonempty $\eonv(S)$ for some $S\subseteq V$ is an edge-cut of $H$. Now we present the minimum distance of the code generated by the incidence matrix of a hypergraph in terms of eonv:

\begin{theorem}\label{main}
Let $H=(V,E)$ be a hypergraph and $M=[m_{ij}]$ be the $n\times m$ vertex-edge incidence matrix of $H$. Let $C=C(H)$ be the binary linear code $\rs(M)$. Then 
$$\d(C)=\min_{\substack{\varnothing \neq S\subseteq V\\ \eonv(S)\neq \varnothing}} |\eonv(S)|.$$
\end{theorem}

\begin{proof}
Let $x\in C$ be a nonzero codeword. Then 
$$x=\sum_{v\in V} \mu_v M_v,$$
where $\mu_v\in \mathbb F_2$ and $M_v$ is the row of $M$ corresponding to vertex $v$ of $H$. Without loss of generality, let 
\[x=M_1+M_2+\cdots+M_k,\]
for some positive integer $k$. Note that $i\in \supp(x)$ if and only if $|e_i\cap \{1,2,\ldots,k\}|$ is odd. In other words, $\supp(x)=\eonv(\{1,2,\ldots,k\})$. Thus

%Thus $w(x)=|\supp(x)|$ is the number of edges that have odd number of common vertices with $\{1,2,\ldots,k\}$. In other words,

$$w_H(x)=|\supp(x)|=|\eonv(\{1,2,\ldots,k\})| \geq \min_{\substack{\varnothing \neq S\subseteq V\\ \eonv(S)\neq \varnothing}} |\eonv(S)|.$$
Therefore
$$\d(C)\geq \min_{\substack{\varnothing \neq S\subseteq V\\ \eonv(S)\neq \varnothing}} |\eonv(S)|.$$

To prove the equality, it suffices to construct a nonzero codeword with weight
$$\min_{\substack{\varnothing \neq S\subseteq V\\ \eonv(S)\neq \varnothing}} |\eonv(S)|.$$
Consider a nonempty set $T\subseteq V$ such that 
$$|\eonv(T)|=\min_{\substack{\varnothing \neq S\subseteq V\\ \eonv(S)\neq \varnothing}} |\eonv(S)|.$$
Suppose 
$$y=\sum_{v\in T} M_v.$$
Since $\supp(y)=\eonv(T)$,
$$w_H(y)=|\supp(y)|=|\eonv(T)|=\min_{\substack{\varnothing \neq S\subseteq V\\ \eonv(S)\neq \varnothing}} |\eonv(S)|.$$
\end{proof}

\begin{corollary}
Let $H=(V,E)$ be a hypergraph and $M=[m_{ij}]$ be the $n\times m$ vertex-edge incidence matrix of $H$. Let $C=C(H)$ be the binary linear code $\rs(M)$. For a codeword $x\in C$ and a set $S\subseteq V$, $x=\sum_{v\in S} M_v$ if and only if $\supp(x)=\eonv(S)$. 
\end{corollary}

\subsection{Complete 3-partite 3-uniform hypergraph} 
Consider the complete $3$-partite 3-uniform hypergraph with partite sets each of size $n$. Suppose $V=X\cup Y\cup Z$ where $X=\{x_1,x_2,\ldots,x_n\}$, $Y=\{y_1,y_2,\ldots,y_n\}$, and $Z=\{z_1,z_2,\ldots,z_n\}$. The edges are all possible triples $\{x_i, y_j, z_k\}$ where $1\leq i, j, k \leq n$. We find a lower bound on the minimum distance of a code generated by the incidence matrix of such a hypergraph, using the result in Theorem \ref{main}. Let us look at a few examples first:

\begin{example}
Let $H$ be the complete $3$-partite 3-uniform hypergraph with partite sets each of size 2. Then the incidence matrix of $H$ is given by 
$$
\left [ \begin{array}{cccccccc}
    1 & 1 & 1& 1& 0&0&0&0\\
    0 & 0 & 0& 0& 1&1&1&1\\
    1&0&0&1&1&0&0&1\\ 
    0&1&1&0&0&1&1&0\\
    1&1&0&0&1&1&0&0 \\
    0&0&1&1&0&0&1&1
    \end{array} \right ],$$
which generates a binary code with parameters $[8,4,4]$.
\end{example}

\begin{example}
Let $H$ be the complete $3$-partite 3-uniform hypergraph with partite sets each of size 3. Then the incidence matrix of $H$ is given by 
$$
\left [ \begin{array}{ccccccccccccccccccccccccccc}
    0& 0& 0& 0& 0& 0& 0& 0& 0& 0& 0& 0& 0& 0& 0& 0& 0& 0& 1& 1& 1& 1& 1& 1& 1& 1& 1\\
0& 0& 0& 0& 0& 0& 0& 0& 0& 1& 1& 1& 1& 1& 1& 1& 1& 1& 0& 0& 0& 0& 0& 0& 0& 0& 0\\
1& 1& 1& 1& 1& 1& 1& 1& 1& 0& 0& 0& 0& 0& 0& 0& 0& 0& 0& 0& 0& 0& 0& 0& 0& 0& 0\\
0& 0& 0& 0& 0& 0& 1& 1& 1& 0& 0& 0& 0& 0& 0& 1& 1& 1& 0& 0& 0& 0& 0& 0& 1& 1& 1\\
0& 0& 0&1& 1& 1& 0& 0& 0& 0& 0& 0& 1& 1& 1& 0& 0& 0& 0& 0& 0& 1& 1& 1& 0& 0& 0\\
1& 1& 1& 0& 0& 0& 0& 0& 0& 1& 1& 1& 0& 0& 0& 0& 0& 0& 1& 1& 1& 0& 0& 0& 0& 0& 0\\
0& 0& 1& 0& 0& 1& 0& 0& 1& 0& 0& 1& 0& 0& 1& 0& 0& 1& 0& 0& 1& 0& 0& 1& 0& 0& 1\\
0& 1& 0& 0& 1& 0& 0& 1& 0& 0& 1& 0& 0& 1& 0& 0& 1& 0& 0& 1& 0& 0& 1& 0& 0& 1& 0\\
1& 0& 0& 1& 0& 0& 1& 0& 0& 1& 0& 0& 1& 0& 0& 1& 0& 0& 1& 0& 0& 1& 0& 0& 1& 0& 0
    \end{array} \right ],$$
which generates a binary code with parameters $[27,7,9]$.
\end{example}

\begin{example}
If $H$ is the complete $3$-partite 3-uniform hypergraph with partite sets each of size 4, then the incidence matrix of $H$ generates a binary code of parameters $[64,10,16]$.
\end{example}

We have the following theorem:
\begin{theorem}
The incidence matrix of the complete $3$-partite 3-uniform hypergraph with partite sets each of size $n$ generates a binary code with minimum distance $n^2$.  
\end{theorem}
\begin{proof}
Let $V=X\cup Y\cup Z$ be the partition of the vertex set. Let $S$ be a non-empty subset of $V$. Since any edge contains exactly three vertices, $\eonv(S)$ will consist of edges that intersect $S$ at one vertex or three vertices. Assume $|S\cap X| = k_1$, $|S\cap Y|=k_2$, and $|S\cap Z| = k_3$. %Then it is easy to calculate $|\eonv(S)|$, 
Note that the number of edges that intersect $S$ in three vertices is given by $k_1k_2k_3$ and the ones that intersect $S$ in exactly one vertex is given by \[k_1(n-k_2)(n-k_3)+(n-k_1)k_2(n-k_3)+(n-k_1)(n-k_2)k_3.\]
Thus we have 
\[|\eonv(S)| = k_1(n-k_2)(n-k_3)+(n-k_1)k_2(n-k_3)+(n-k_1)(n-k_2)k_3+k_1k_2k_3 =:f(k_1, k_2, k_3).\]
Now we find the minimum positive value of $f(k_1, k_2, k_3)$ for integers $k_1,k_2,k_3$ over the region $0\leq k_1 \leq n, 0\leq k_2\leq n$, and $0\leq k_3\leq n$, with $k_1+k_2+k_3 \neq 0$.

We first note that $f(k_1, k_2, k_3) = 0$ if and only if $k_1=k_2=k_3=0$ (i.e., $S=\varnothing$) or $(k_1, k_2, k_3) = (n, n,0)$, $(n,0,n)$ or $(0,n,n)$. The next observation is that $f(1,0,0) = f(0,1,0) = f(0,0,1) = n^2$.  We show $|\eonv(S)| =f(k_1, k_2, k_3)\geq n^2$ in the following three cases:\\

{\bf Case 1:} Exactly two of $k_1, k_2, k_3$ is zero.\\
We look at $f(k_1, 0, 0)$ since the expression is symmetric. 
$$f(k_1, 0,0) = k_1n^2 \geq n^2, \:  1\leq k_1\leq n.$$

{\bf Case 2:} Exactly one of $k_1, k_2, k_3$ is zero.\\
By symmetry, we just look at $f(k_1, k_2, 0)$:
$$f(k_1, k_2, 0) = k_1n(n-k_2)+k_2n(n-k_1) = n[(n-k_2)k_1+(n-k_1)k_2].$$
Since we do not want $f(k_1,k_2,0)$ to be zero, we assume $k_1 k_2 <n^2$. 

If $k_1<n$ and $k_2=n$, then the expression reduces to $n^2(n-k_1)\geq n^2$, since $n-k_1\geq 1$.

If $k_1,k_2<n$, then we have $k_2, n-k_2 \geq 1$ and so we get $(n-k_2)k_1+(n-k_1)k_2 \geq k_1+(n-k_1)\geq n$ and hence
$n[(n-k_2)k_1+(n-k_1)k_2] \geq n^2$. \\

{\bf Case 3:} $k_1, k_2, k_3$ are all non-zero. \\
We rewrite the function as 
\begin{align*}
f(k_1,k_2,k_3) & = k_1(n-k_2)(n-k_3)+(n-k_1)k_2(n-k_3)+(n-k_1)(n-k_2)k_3+k_1k_2k_3 \\
& = k_3[(n-k_1)(n-k_2)+k_1k_2]+(n-k_3)[k_1(n-k_2)+(n-k_1)k_2] \\
& = k_3g(k_1,k_2)+(n-k_3)h(k_1, k_2),
\end{align*}

where $1\leq k_1, k_2, k_3 \leq n$, $g(k_1,k_2) = (n-k_1)(n-k_2)+k_1k_2$, and $h(k_1, k_2) = k_1(n-k_2)+(n-k_1)k_2$.

If $k_1=k_2=n$, then $g(k_1, k_2) = n^2$ and $h(k_1, k_2)=0$ and so $f(k_1, k_2, k_3) = k_3n^2\geq n^2$.

If $k_1=n$ and $k_2 < n$, then $g(k_1, k_2) = nk_2 \geq n$ and $h(k_1,k_2) = n(n-k_2) \geq n$, in which case $f(k_1,k_2,k_3) \geq k_3n+(n-k_3)n = n^2$.

Finally, if $k_1, k_2 <n$, then $g(k_1,k_2)\geq n-k_1+k_1 = n$ and $h(k_1,k_2) \geq n-k_2+k_2 = n$, which results in $f(k_1,k_2,k_3) \geq k_3n+(n-k_3)n = n^2$.

\end{proof}

\subsection{A special case: projective plane} A projective plane $PG(n-1,2)$ is a combinatorial structure that consists of points, which are 1-dimensional subspaces of $\F_2^n$ and lines, which are 2-dimensional subspaces. We can consider the projective plane to be a hypergraph by taking the points to be vertices and the lines to be the edges. The following theorem is presented in terms of eonv from its original (Lemma 3.1) that appeared in \cite{baartmans}. 
\begin{theorem}
Let $A$ be a set of points in $PG(n-1,2)$, $n\geq 3$. 
\begin{enumerate}
    \item If $|\eonv(A)|=0$, then $A=\varnothing$ or $A$ is a hyperoval.
    \item If $|\eonv(A)|>0$, then $|\eonv(A)|\geq 2^{n-1}-1$ with the equality only if $A$ is a single point, or the union of a point and a hyperoval, or a hyperoval with a point deleted. 
    \item If $|\eonv(A)|>2^{n-1}$, then $|\eonv(A)\geq 2^n-4$ with the equality only if $A$ is a set of two points, or $A$ is the union of a hyperoval with $i$ points deleted and $2-i$ points off the hyperoval, $i=0,1,2$. 
\end{enumerate}
\end{theorem}

Combining this with Theorem \ref{main} we get the following corollary:
\begin{corollary}
If $C$ is the binary code generated by the incidence matrix of $PG(n-1,2)$, then the minimum distance $d(C)$ of the code satisfies 
$$d(C)\geq \min\{2^{n-1}-1, 2^n-4\}.$$
\end{corollary}

\begin{example}
The Fano plane is the projective plane $PG(2,2)$.  The points are all the 1-dimensional subspaces of $\F_2^3$, which are all generated by the non-zero vectors in $\F_2^3$. A two-dimensional subspace will have three non-zero vectors from $\F_2^3$.
Letting $0 \mapsto (0,0,1) $, $1\mapsto (0,1,0)$, $3\mapsto (0,1,1)$, $2\mapsto (1,0,0)$, $4 \mapsto (1,1,0)$, $5 \mapsto (1,1,1)$, $6 \mapsto (1,0,1)$, and letting $V=\{0,1,2,3,4,5,6\}$, we get the lines to be the following set of triples:
$$\{0,1,3\}, \{1,2,4\}, \{2,3,5\},\{3,4,6\},\{4,5,0\}, \{5,6,1\},\{6,0,2\}.$$
The binary code corresponding to the Fano plane will be generated by the following matrix, if we take the above ordering for the lines:
$$
 \begin{bmatrix}
    1 & 0 & 0 & 0 & 1 & 0 & 1\\
    1 & 1 & 0 & 0 & 0 & 1 & 0 \\
    0 & 1 & 1 & 0 & 0 & 0 & 1 \\
    1 & 0 & 1 & 1 & 0 & 0 & 0 \\
    0 & 1 & 0 & 1 & 1 & 0 & 0 \\
    0 & 0 & 1 & 0 & 1 & 1 & 0 \\
    0 & 0 & 0 & 1 & 0 & 1 & 1
    \end{bmatrix}.
$$
The binary code generated by this matrix turns out to be equivalent to the well known Hamming code of parameters $[7,4,3]$. As we see, the minimum distance in this case is $2^2-1 = 2^{n-1}-1$. 

\begin{figure}[h]
	\begin{center}
	\begin{tikzpicture}[scale=1.5, colorstyle/.style={circle, fill, black, scale = .5}, >=stealth]
		
		\node (2) at (0,1.85)[colorstyle, label=above:$0$]{};
		\node (7) at (0,0)[colorstyle, label=below:$5$]{};
		\node (1) at (-1.8,-1)[colorstyle, label=below:$1$]{};
		\node (4) at (1.8,-1)[colorstyle, label=below:$2$]{};	
		\node (5) at (0,-1)[colorstyle, label=below:$4$]{};	
		\node (6) at (0.85,0.5)[colorstyle, label=right:$6$]{};	
		\node (3) at (-0.85,0.5)[colorstyle, label=left:$3$]{};	
		
		\draw [] (2)--(7)--(5);
        \draw [] (1)--(7)--(6);
        \draw [] (1)--(3)--(2);
        \draw [] (1)--(5)--(4);
        \draw [] (4)--(6)--(2);
        \draw [] (4)--(7)--(3);
 	\draw[thick](0,0) circle (0.98);
	\end{tikzpicture}
	\caption{The Fano plane}\label{Fano}
	\end{center}
	\end{figure}
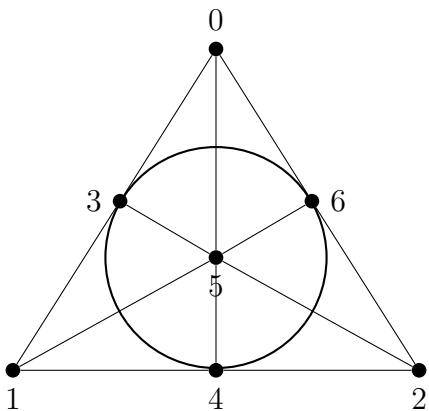
\end{example}

\subsection{Circulant hypergraphs and cyclic codes}
Consider a circulant matrix, that is a matrix given by 
$$
M=\begin{bmatrix}
    a_1    & a_2    & a_3   & \hdots & a_n \\
    a_n    & a_1    & a_2   & \hdots & a_{n-1} \\
    a_{n-1}& a_n    & a_1   & \hdots & a_{n-2} \\
    \vdots & \vdots & \vdots& \ddots & \vdots \\
    a_2    & a_3    & a_4   & \hdots & a_1
\end{bmatrix}. 
$$
Note that the matrix is uniquely determined by the first row, where every other row is a cyclic shift of the previous row. If we consider such a circulant matrix to be the incidence matrix for a hypergraph, we see that the hypergraph is both vertex-regular and edge-regular with the same degree of regularity for both. We call such a hypergraph a ``circulant hypergraph". Consider a circulant hypergraph $H=(V,E)$ on $n$ vertices $1,2,\ldots,n$ and $m$ edges $e_1,e_2,\ldots,e_m$ whose incidence matrix is $M$ with rows $M_1,M_2,\ldots,M_n$. Suppose $C$ is the code $C=\rs(M)$. 

Now we describe the code $C=\rs(M)$ in terms of polynomials. For a single-variable polynomial $p(x)$, its support is \[\supp(p):=\text{ the set of indices of $x$ (with nonzero coefficients) in $p(x)$}.\]
For example, if $p(x)=5-x^2$, then $\supp(p)=\{0,2\}$. If $p\in \mathbb F_2(x)$, then  $p(x)=\displaystyle\sum_{j\in \supp(p)} x^j$.

Note that $C=\rs(M)$ is generated by polynomials $p_1,p_2,\ldots,p_n$ where 
\[p_i(x)=\displaystyle\sum_{j\in \supp(M_i)} x^{j-1}.\]
Note that $p_i$ corresponds to vertex $i$ in $H$ for all $i=1,2,\ldots,n$. Then $C$ is generated by $\{p_{1},p_{2},\ldots,p_{n}\}$.\\

For a set $S\subseteq V=\{1,2,\ldots,n\}$, $\eonv(S)=\left\lbrace e_{i+1}\;:\;i\in \supp\left(\displaystyle\sum_{t\in S} p_{t}\right)\right\rbrace$ and 
\[|\eonv(S)|=\left\vert\supp\left(\displaystyle\sum_{t\in S} p_{t}\right)\right\vert.\]
Then by Theorem \ref{main},
\[\d(C)=\min_{\substack{\varnothing \neq S\subseteq V\\ \sum_{t\in S} p_{t}\neq 0}} \left\vert\supp\left(\displaystyle\sum_{t\in S} p_{t}\right)\right\vert.\]

We can actually say more about the $2$-rank of incidence matrix $M$, by exploring the connection that corresponding  hypergraph $H$ have with cyclic codes, i,e, a code $C$ that is invariant under the cyclic shift. More precisely, if $M$ is a generator matrix for a binary linear code $C$, then $C$ is a cyclic code generated by $p(x)=a_1+a_2x+\dots +a_{n-1}x^{n-1}$ over the quotient ring $\F_2[x]/\langle x^n-1\rangle$. The $2$-rank of $M$ is the same as the dimension of the binary cyclic code generated by $p(x)$ of length $n$. 

\begin{theorem} $($\cite{macwilliams}$)$
Let $C$ be the binary cyclic code generated by $p(x)$ of length $n$. If $g(x) = GCD(p(x), x^{n}-1)$ in $\F_2[x]$, then 
$$dim(C) = n-deg(g(x)).$$
\end{theorem}

Now we explore the minimum distance of certain cyclic codes generated by circulant hypergraphs. Consider a circulant hypergraph $H=(V,E)$. For an hyperedge $e$, $\overline{e}$ denotes the complement of the hyperedge $e$, i.e., $\overline{e}=V\setminus e$ If $e$ corresponds to the binary column $\mathbf{c}$ in the incidence matrix of $H$, then $\overline{e}$ corresponds to the binary vector $\overline{\mathbf{c}}$ which is given by \[\overline{\mathbf{c}}=1+\mathbf{c} \pmod{2}.\]

\begin{lemma}\label{lem}
Consider a circulant hypergraph $H=(V,E)$. Let $S$ be a subset of $V$ and $e$ be a hyperedge. Then we have the following:
\begin{enumerate}
    \item[(a)] If $|S|$ is even, then $|S\cap e|$ is odd if and only $|S\cap \overline{e}|$ is odd. 
    \item[(b)] If $|S|$ is odd, then $|S\cap e|$ is odd if and only if $|S\cap \overline{e}|$ is even. 
\end{enumerate}

\end{lemma}
\begin{proof}
Both parts follow from the following simple observation:
$S=(S\cap e)\cup (S\cap \overline{e})$ and since both these sets are disjoint, we have 
$$|S| = |S\cap e|+|S\cap \overline{e}|.$$
\end{proof}

Consider the $km$-uniform  hypergraph on $2km$ vertices whose incidence matrix is the circulant matrix with the first row $k$ blocks of $11..100...0$ ($m$ $1$s followed by $m$ zeroes).
We first look at two examples of such matrices, one with $m=2$ and $m=3$.

The following is an example with $m=2$.
$$ M=\left [ \begin{array}{cccccccccccc}
   1 &1& 0& 0& 1& 1& 0& 0& 1& 1& 0& 0\\ 
   0& 1& 1& 0& 0& 1& 1& 0& 0& 1 & 1& 0\\
  0 & 0& 1 & 1& 0& 0& 1& 1& 0& 0& 1& 1\\ 
  1& 0& 0& 1& 1& 0& 0& 1& 1& 0& 0& 1 \\
  1& 1& 0& 0& 1& 1& 0& 0& 1& 1& 0& 0 \\
  0& 1& 1& 0& 0& 1& 1& 0& 0& 1& 1& 0\\
  0& 0& 1& 1& 0& 0& 1& 1& 0& 0& 1& 1\\
  1& 0& 0& 1& 1& 0& 0& 1& 1& 0& 0& 1\\
  1& 1& 0& 0& 1& 1& 0& 0& 1& 1& 0& 0\\
  0& 1&1& 0& 0& 1& 1& 0& 0& 1& 1& 0\\
  0& 0& 1& 1& 0& 0& 1& 1& 0& 0& 1&1 \\
1& 0& 0& 1& 1& 0& 0& 1& 1& 0& 0& 1
 \end{array} \right]. 
$$
Note that the columns of $M$ can be labeled as
$$\{\mathbf{c}_1, \mathbf{c}_2,\overline{\mathbf{c}_1}, \overline{\mathbf{c}_2}, \mathbf{c}_1, \mathbf{c}_2,\overline{\mathbf{c}_1}, \overline{\mathbf{c}_2}\},$$ which corresponds to
$$\{e_1, e_2, \overline{e_1}, \overline{e_2}, e_1, e_2, \overline{e_1}, \overline{e_2}\}$$ as hyperedges.  

The following is an example with $m=3$.
$$
M= \left[ \begin{array}{cccccccccccc}1& 1& 1& 0& 0& 0& 1& 1& 1& 0& 0& 0 \\ 0& 1& 1& 1& 0& 0& 0& 1& 1& 1& 0& 0\\ 0& 0& 1& 1& 1& 0& 0& 0& 1& 1& 1& 0\\ 0& 0& 0& 1& 1& 1& 0& 0& 0& 1& 1& 1\\ 1& 0& 0& 0& 1& 1& 1& 0& 0& 0& 1& 1\\ 1& 1& 0& 0& 0& 1& 1& 1& 0& 0& 0& 1\\ 1& 1& 1& 0& 0& 0& 1& 1& 1& 0& 0& 0 \\ 0& 1& 1& 1& 0& 0& 0& 1& 1& 1& 0& 0 \\ 0& 0& 1& 1& 1& 0& 0& 0& 1& 1& 1& 0 \\ 0& 0& 0& 1& 1& 1& 0& 0& 0& 1& 1& 1 \\ 1& 0& 0& 0& 1& 1& 1& 0& 0& 0& 1& 1 \\ 1& 1& 0& 0& 0& 1& 1& 1& 0& 0& 0& 1 \end{array} \right]. 
$$
In this case, the columns of $M$ can be labeled as 
$$\{\mathbf{c}_1, \mathbf{c}_2,\mathbf{c}_3, \overline{\mathbf{c}_1}, \overline{\mathbf{c}_2}, \overline{\mathbf{c}_3},\mathbf{c}_1, \mathbf{c}_2,\mathbf{c}_3, \overline{\mathbf{c}_1}, \overline{\mathbf{c}_2}, \overline{\mathbf{c}_3}, \},$$ which corresponds to
$$\{e_1, e_2, e_3, \overline{e_1}, \overline{e_2}, \overline{e_3}, e_1, e_2, e_3, \overline{e_1}, \overline{e_2}, \overline{e_3}\}$$
as hyperedges.

\begin{theorem}
Let $\mathbf{r}$ be a binary vector of length $2km$, where there are $k$ blocks of $11..100...0$ ($m$ $1$s followed by $m$ zeroes). Then the minimum distance of the cyclic code generated by $\mathbf{r}$ is at least $k$ if $m=1$ and it is at least $2k$ if $m\geq 2$. 
\end{theorem}

\begin{proof}
Let $C$ be the cyclic code generated by $\mathbf{r}$ and $H=(V,E)$ be the corresponding circulant hypergraph with the circulant incidence matrix $M$ whose first row is $\mathbf{r}$.

We first look at the special case of $m=1$ in which every column of $M$ is either equal to 
$$\mathbf{c}_i = 
\left [ \begin{array}{c}
    1\\
    0 \\
    1  \\
    0 \\
    . \\
    .  \\
    .  \\
    1\\
    0
    \end{array} \right ]
 \:\:\: \textrm{or} \:\:\: \left [ \begin{array}{c}
    0\\
    1 \\
    0  \\
    1 \\
    . \\
    .  \\
    .  \\
    0\\
    1
    \end{array} \right ].$$
Since we have $k$ copies of each, this means that for any $S\subseteq V$, we have $|\eonv(S)| = 0, k$ or $2k$. This proves the case for $m=1$ by Theorem \ref{main}.

Now assume $m\geq 2$. We prove that $|\eonv(S)| =0$ or $ \geq 2k$ for any non-empty subset $S$ of the vertex set $V$.  

The main observation is that because of the block structure of the rows, $\mathbf{c}_i$ shifted $m$ times will result in $\overline{\mathbf{c}_i}$ and $\mathbf{c}_i$ shifted $2m$ times will result in $\mathbf{c}_i$ again. So, all the columns of $M$ can be expressed as $k$ copies of $\{\mathbf{c}_1, \overline{\mathbf{c}_1}, \mathbf{c}_2, \overline{\mathbf{c}_2}, \dots, \mathbf{c}_k, \overline{\mathbf{c}_k} \}$ (see preceding examples), with the corresponding hyperedges denoted by $\{e_1, \overline{e_1},e_2, \overline{e_2}, \dots, e_m, \overline{e_m} \}$. 

Assume that $|S|$ is even. Then, by Lemma \ref{lem}, $|S\cap e_i|$ is odd if and only if $|S\cap \overline{e_i}|$ is odd. This means, for any such $S$, we have two cases:
\begin{itemize}
 \item $|S\cap e_i|$ is even for all $i=1,2,\ldots,m$. In this case, $\eonv(S)=\varnothing$ for all $i=1,2,\ldots,m$.
\item  $|S\cap e_i|$ is odd for at least one $i$. In this case, we have $|\eonv(S)| \geq 2k$ since all $k$ copies of $e_i$ and $\overline{e_i}$ will be counted in $\eonv(S)$. 
\end{itemize}

Now assume $|S|$ is odd.  Then, by Lemma \ref{lem}, either $|S\cap e_i|$ is odd or $|S\cap \overline{e_i}|$ is odd for all $i=1, 2, \dots, m$. Thus we have $|\eonv(S)| \geq km \geq 2k$ since $m\geq 2$. 

Finally the result follows from Theorem \ref{main}: 
\[\d(C)=\min_{\substack{\varnothing \neq S\subseteq V\\ \eonv(S)\neq \varnothing}} |\eonv(S)|\geq 2k.\]
\end{proof}

\section{Open problems regarding self-duality}
% of codes from hypergraphs}
Let $C$ be a binary code generated by the incidence matrix $M$ of a hypergraph. Since $C=\rs(M)$, $C^{\perp}=\ns(M)$, the null space of $M$. Recall that $C$ is self-orthogonal if $C\subseteq C^{\perp}$ and $C$ is self-dual if $C= C^{\perp}$. Therefore $C$ is self-dual if and only if $\rs(M)=\ns(M)$. Also, $C$ is self-dual if it is self-orthogonal and its dimension is $n/2$, where $n$ is the length of the code or the number of columns of $M$. 

\begin{remark}\label{inc} 
Let $M$ be the incidence matrix of a connected graph $G$ on $n$ vertices. Consider  the linear code $C=\rs(M)$ generated by $M$.

\begin{enumerate}
    \item The dimension of $C=\rs(M)$ is either $n-1$ or $n$, which means that $G$ cannot generate a self-dual code if the number of edges of $G$ is odd or less than $2n-2$. 
    
    \item $C=\rs(M)$ is self-orthogonal if and only if $MM^T=O_n$ in $\mathbb F_2$ if and only if $M_u\cdot M_v=0$ in $\mathbb F_2$ for all vertices $u,v$ in $G$, which is equivalent to 
\begin{enumerate}
    \item even degree of each vertex, and
    \item $|\e[u]\cap \e[v]|$ is even for all distinct vertices $u,v$, where $\e[u]$ denotes the set of edges containing $u$.
\end{enumerate}
\end{enumerate}

\end{remark}

For the special case of a connected graph, we have the following theorem:
\begin{theorem}\label{selfdual}
Let $G$ be a connected graph on $n$ vertices and $m$ edges with incidence matrix $M$. Then the binary linear code $\rs(M)$ is self-dual if and only if
\begin{enumerate}
    \item[(a)] $m=2n-2$, and
    \item[(b)] $|\e(u)\cap \e(v)|$ is even for all vertices $u,v$ in $G$ (which implies even degree of each vertex).
\end{enumerate}
\end{theorem}
\begin{proof}
The condition (b) is necessary and sufficient for $RS(M)$ to be self-orthogonal by Remark \ref{inc}.
Since $G$ is a connected graph on $n$ vertices, the $2$-rank of $RS(M)$ is $n-1$ (see \cite{dankelmann1}).
Then $RS(M)$ generates a self-dual code if and only if half of the length $m$ of the code is $n-1$, i.e., $m=2n-2$. 
\end{proof}

The preceding result is for a connected 2-uniform hypergraph. A natural question would be to extend the result for connected k-uniform hypergraphs:

\begin{question}
Let $G$ be a connected $k$-uniform hypergraph with incidence matrix $M$. Find necessary and sufficient conditions for self-duality of the binary linear code $\rs(M)$. 
\end{question}

The answer to the above question will depend on the rank of the incidence matrix $M$ finding which is not easy for $k$-uniform hypergraphs when $k>2$.

\end{document}